\documentclass[a4paper,12pt,reqno]{amsart}
\usepackage{amsmath,amsthm,amsfonts,amssymb}
\usepackage{indentfirst}
\usepackage{a4wide}
\usepackage{caption}
\usepackage{microtype}
\usepackage{xcolor}
\usepackage{xspace}
\usepackage{graphicx}
\usepackage{hyperref}

\theoremstyle{plain}
\newtheorem{theorem}{Theorem}[section]
\newtheorem{corollary}[theorem]{Corollary}
\newtheorem{lemma}[theorem]{Lemma}

\theoremstyle{definition}

\theoremstyle{remark}
\newtheorem*{remark}{Remark}

\numberwithin{equation}{section}
\numberwithin{figure}{section}

\hypersetup{colorlinks}

\newcommand{\bbR}{\mathbb{R}}
\newcommand{\bbZ}{\mathbb{Z}}

\newcommand{\bbQP}{\mathbb{Q}_{+}}
\newcommand{\eps}{\varepsilon}

\newcommand{\alsu}{\mathrm{a.s.}\xspace}

\newcommand{\al}{\alpha}
\newcommand{\la}{\lambda}
\newcommand{\te}{\theta}
\newcommand{\tu}{\theta_1}
\newcommand{\td}{\theta_2}
\newcommand{\ga}{\gamma}

\newcommand{\xf}{x_{\infty}}

\newcommand{\tf}{\tau_{\infty}}

\newcommand{\ap}{(\alpha, p)}
\newcommand{\gt}{(\gamma, \theta)}
\newcommand{\vapdk}{\((\alpha, p)\)-probability variant of the Daley--Kendall model\xspace}
\newcommand{\vapmt}{\((\alpha, p)\)-probability variant of the Maki--Thompson model\xspace}
\newcommand{\apdk}{\((\alpha, p)\) Daley--Kendall model\xspace}
\newcommand{\apmt}{\((\alpha, p)\) Maki--Thompson model\xspace}
\newcommand{\dk}{Daley--Kendall model\xspace}
\newcommand{\mt}{Maki--Thompson model\xspace}
\newcommand{\pmax}{M^{(N)}\xspace}
\newcommand{\tpmax}{\tilde{M}^{(N)}\xspace}
\newcommand{\lmax}[1]{m^{\star}_{\mathrm{#1}}\xspace}
\newcommand{\tabs}{\tau^{(N)}\xspace}
\newcommand{\ttabs}{\tilde{\tau}^{(N)}\xspace}

\newcommand{\LW}[1]{\ensuremath{W_0(#1)}}

\begin{document}

\title[The maximum proportion of spreaders]{The maximum proportion of spreaders\\in stochastic rumor models}

\author[E. Lebensztayn]{Elcio Lebensztayn}
\address[E. Lebensztayn]{Instituto de Matem\'atica, Estat\'istica e Computa\c{c}\~ao Cient\'ifica, Universidade Estadual de Campinas (UNICAMP), CEP 13083-859, Campinas, SP, Brasil.}
\email{lebensz@unicamp.br}
\author[P. M. Rodriguez]{Pablo M. Rodriguez}
\address[P. M. Rodriguez]{Centro de Ci\^encias Exatas e da Natureza, Universidade Federal de Pernambuco (UFPE), CEP 50670-901, Recife, PE, Brasil.}
\email{pablo@de.ufpe.br}


\date{\today}

\begin{abstract}
We examine a general stochastic rumor model characterized by specific parameters that govern the interaction rates among individuals.
Our model includes the \((\alpha, p)\)-probability variants of the well-known Daley--Kendall and Maki--Thompson models.
In these variants, a spreader involved in an interaction attempts to transmit the rumor with probability \(p\); if successful, any spreader encountering an individual already informed of the rumor has probability \(\alpha\) of becoming a stifler.
We prove that the maximum proportion of spreaders throughout the process converges almost surely, as the population size approaches~\(\infty\).
For both the classical Daley--Kendall and Maki--Thompson models, the asymptotic proportion of the rumor peak is \(1 - \log 2 \approx 0.3069\).
\end{abstract}

\maketitle

\section{Introduction}
\label{S: Introduction}

The phenomenon of rumor spreading is significant within the applied social sciences and has implications across various fields, including Economics, Politics, and Public Health. Thus, it is not surprising that its theoretical or computational modeling is currently an active area of research; see, for example, \cite{di-crescenzo-24,di-crescenzo, esmaeeli,mahmoud,paraggio,MRR,san-martin}. From a mathematical perspective, one of the earliest references is the work of Daley and Kendall~\cite{DK}. 
The authors present a pioneering analysis of rumor processes and point out the main differences between these processes and epidemic models. 
Generally, for epidemiological systems (such as the SIR model), two infected individuals do not influence each other's infection.
In contrast, a fundamental assumption in rumor modeling is that an individual who is aware of the rumor will continue to spread it until encountering someone else who also knows the rumor. 
At that moment, the individual transmitting the rumor perceives that there is no longer a reason to continue sharing it.
Besides that, as opposed to epidemic models, rumor processes do not exhibit threshold behavior.

To describe the basic version of the model proposed by Daley and Kendall~\cite{DK}, we consider a closed homogeneously mixing population of \(N + 1\) individuals. 
At any time \(t \geq 0\), people are classified into three disjoint categories, consisting of:
\begin{itemize}
\item \emph{ignorant individuals}, who are not aware of the rumor;
\item \emph{spreader individuals}, who are disseminating the rumor;
\item \emph{stifler individuals}, who know the rumor but have stopped communicating it;
\end{itemize}
and the number of individuals in each category in $t$ is indicated by $X(t)$, $Y(t)$, $Z(t)$, respectively. Initially, \(X(0) = N\), \(Y(0) = 1\) and \(Z(0) = 0\), and \(X(t) + Y(t) + Z(t) = N + 1\) for all~\(t \geq 0\).
Of course, the random variables \(X(t)\), \(Y(t)\), \(Z(t)\) depend on \(N\); this dependence is explicitly displayed only when necessary.

In the \dk, the rumor is propagated through {\it pairwise contacts} between individuals, so that \(\{(X(t), Y(t))\}_{t \geq 0}\) is a continuous-time Markov chain with transitions and corresponding rates given by 
\begin{equation}\label{eq:transDK}
{\allowdisplaybreaks
\begin{array}{ccc}
\text{interaction} \qquad &\text{transition} \qquad &\text{rate}\\[0.1cm]
\text{spreader -- ignorant} \qquad &(-1, 1) \qquad &X Y,\\[0.1cm]
\text{spreader -- spreader} \qquad &(0, -2) \qquad &\displaystyle\binom{Y}{2},\\[0.36cm]
\text{spreader -- stifler} \qquad &(0, -1) \qquad &Y (N + 1 - X - Y).
\end{array}}%
\end{equation}
In the first case, a spreader tells the rumor to an ignorant person, who then becomes a spreader. 
The other two transitions correspond to spreader--spreader and spreader--stifler meetings, resulting in the transformation of the involved spreader(s) into stifler(s). 
This change is due to the loss of interest in transmitting the rumor, stemming from the realization that the other individual is already aware of it. 
Thus defined, \eqref{eq:transDK} means that if the process is in state $(i,j)$ at time $t$, then the probabilities that it jumps to states $(i-1,j+1)$, $(i,j-2)$ or $(i,j-1)$ at time $t+h$ are, respectively, $ijh + o(h)$, $\binom{j}{2}h + o(h)$ and $j(N+1-i-j)h + o(h)$, where $o(h)$ represents a function such that $\lim_{h\to 0}o(h)/h =0$.

Daley and Kendall~\cite[Section~7]{DK} also proposed a more sophisticated model, by assuming that, at each interaction, a spreader tells the rumor with probability \(p\) and that an eventual transformation of a spreader into a stifler occurs with probability \(\alpha\).
In Daley and Gani~\cite{DG}, this process is referred to as the \emph{\vapdk}.
A detailed description is presented in Section~\ref{SS: apdk} of the paper.

Another notable rumor model that has gained popularity, primarily for its simplicity, was formulated by Maki and Thompson~\cite{MT}.
In this model, the rumor is spread through \emph{directed contacts} between spreaders and other individuals, meaning there is a distinction between the initiator and the recipient. 
Moreover, when a spreader interacts with another spreader, only the initiating spreader becomes a stifler.
Hence, the continuous-time Markov chain \(\{(X(t), Y(t))\}_{t \geq 0}\) has the following increments and rates:
\begin{equation*}
\begin{array}{ccc}
\text{interaction} \qquad &\text{transition} \qquad &\text{rate}\\[0.2cm]
\text{spreader -- ignorant} \qquad &(-1, 1) \qquad &X Y,\\[0.2cm]
\text{spreader -- spreader/stifler} \qquad &(0, -1) \qquad &Y (N - X).
\end{array}
\end{equation*}
In this case, note that if the process is in state $(i,j)$ at time $t$, then the probabilities that it jumps to states $(i-1,j+1)$ or $(i,j-1)$ at time $t+h$ are, respectively, $ijh + o(h)$ and $j (n - i)h + o(h)$.

Much of the literature on stochastic rumor models focuses on analyzing the outcome of the process.
By using martingales, Sudbury~\cite{Sudbury} established a Law of Large Numbers for the \mt, stating that the proportion of individuals who ultimately have not heard the rumor converges in probability to \(0.203\) as \(N \to \infty\). 
The asymptotic normality of this proportion, when suitably rescaled, was proved by Watson~\cite{Watson} for the \mt and by Pittel~\cite{Pittel} for the \dk.
Lebensztayn~\cite{LDPMT} derived a closed formula for the probability mass function of the final number of ignorants in the \mt, along with a Large Deviations Principle for the corresponding proportion.
For a general stochastic rumor system that includes both Daley--Kendall and Maki--Thompson models, Pearce~\cite{Pearce} examined the probability generating functions through a block-matrix approach. 
Lebensztayn et al.~\cite{LTRM} also considered a general rumor model and demonstrated a Law of Large Numbers and a Central Limit Theorem for the outcome of the process; their model incorporates a new category of uninterested individuals.
In addition, limit theorems were proved by Lebensztayn et al.~\cite{RPRS}, for a generalization of the \mt allowing a spreader a random number of unsuccessful telling interactions.
Rada et al.~\cite{MRR} investigated the asymptotic behavior of a generalized \mt, assuming that each ignorant becomes a spreader only after hearing the rumor a predetermined number of times.
The interested reader is referred to Daley and Gani~\cite[Chapter~5]{DG} for a comprehensive overview of the theory of rumor models.

In this paper, we consider a stochastic rumor model that encompasses the basic Daley--Kendall and Maki--Thompson models, and their \((\alpha, p)\)-probability variants.
This generalized model, whose definition is given in Section~\ref{SS: GRM}, is characterized by a set of suitably chosen parameters that dictate the transition rates of the process.
This construction enables a quantitative framework for describing the behavioral mechanisms of the people involved in rumor spreading.
For this rumor model, we define
\[ \tabs = \inf \{t: Y^{(N)}(t) = 0\} \]
the random time when there are no spreaders left in the population. It is worth noting that the superscript $(N)$ is used whenever we want to emphasize the dependence on~$N$ of the random variables. Our main purpose is to study the asymptotic behavior as \(N \to \infty\) of the random variable
\begin{equation}
\label{F: PMax}
\pmax = \max_{0 \leq t \leq \tabs} \frac{Y^{(N)}(t)}{N},
\end{equation}
which is the maximum proportion of spreaders during the entire process.
We prove that \(\pmax\) converges almost surely as \(N \to \infty\) to a limiting constant depending on the model parameters.
To the best of our knowledge, limit theorems for the highest proportion of spreaders attained in rumor models have not been established previously. Our proofs, presented in Section \ref{S: Proofs}, rely on a suitable application of convergence results for density-dependent Markov chains. The main idea is to show that, as $N$ goes to infinity, the entire trajectories of a suitably coupled version of our model, rescaled by $N$, converge to the solution of a tractable system of differential equations. The key point is that, after applying a random time change, we obtain a new process that has the same transitions as the original one, so their corresponding maximum numbers of spreaders until absorption are the same. Moreover, the coupling between the two processes is constructed in such a way that the new process becomes a density-dependent Markov chain, allowing us to apply the theory of convergence for such chains.
For a thorough treatment of this theory, we refer the reader to Andersson and Britton \cite[Chapter~5]{AB} and Ethier and Kurtz \cite[Chapter~11]{MPCC}.

\section{Main results}
\label{S: Main results}

\subsection{The \((\texorpdfstring{\alpha}{alpha}, p)\) Daley--Kendall model}
\label{SS: apdk}

First, we consider the \emph{\vapdk}, which we shall abbreviate as \apdk.
Given \(\alpha \in (0, 1]\) and \(p \in (0, 1]\), assume that, independently for each pairwise meeting and each individual,
\begin{itemize}
\item[(a)] A spreader involved in a meeting chooses to transmit the rumor with probability~\(p\).

\item[(b)] Once this decision is made, any spreader in a meeting with someone aware of the rumor becomes a stifler with probability~\(\al\).

\end{itemize}
Thus, if $X(t)$ and $Y(t)$ denote the number of ignorants and spreaders at time $t$, then the \apdk is the continuous-time Markov chain \(\{(X(t), Y(t))\}_{t \geq 0}\) with possible increments and corresponding infinitesimal rates given by

\begin{equation}\label{F: Rates apdk}
{\allowdisplaybreaks
\begin{array}{ccc}
\text{interactions}\quad&\text{transition} \quad &\text{rate}\\[0.2cm]
\text{spreader -- ignorant}\quad&(-1, 1) \quad &p \, X Y,\\[0.2cm]
\text{spreader -- spreader}\quad&(0, -2) \quad &\al^2 p (2 - p) \displaystyle\binom{Y}{2},\\[0.4cm]
\text{spreader -- spreader/stifler}\quad&(0, -1) \quad &\al (1 - \al) p (2 - p) \, Y (Y - 1) + \al p \, Y (N + 1 - X - Y).
\end{array}}%
\end{equation}

\noindent
The first equation of \eqref{F: Rates apdk} represents the case in which a spreader meets an ignorant and chooses to transmit the rumor, which occurs with probability $p$; and as a result, the ignorant becomes a spreader. The second equation of \eqref{F: Rates apdk} corresponds to a meeting between two spreaders, where at least one of them transmits the information---an event that occurs with probability $p(2-p)$. After the transmission, both spreaders become stiflers, which happens with probability $\alpha^2$. Finally, the third equation of \eqref{F: Rates apdk} summarizes two possible interactions that result in a spreader changing category and becoming a stifler. The first case is when two spreaders meet, at least one of them chooses to transmit the information, and afterward, only one of them becomes a stifler; this last event occurs with probability $\alpha(1-\alpha)$. The second case is when a spreader meets a stifler, the spreader transmits the information, and then becomes a stifler; this event occurs with probability $\alpha p$.

\begin{theorem}
\label{T: apdk}
Consider the \apdk with \(\alpha, p \in (0, 1]\), and let \(\pmax\) be as in \eqref{F: PMax}. Then,
\[ \lim_{N \to \infty} \pmax = \lmax{DK}\ap \quad \alsu, \]
where
\[ \lmax{DK}\ap = 
\left\{
\begin{array}{cl}
\dfrac{1}{\alpha (1 - p)} \left[ (1 + \alpha) \left(\dfrac{\alpha (2 - p)}{1 + \alpha}\right)^{\frac{1}{1 - \alpha (1 - p)}} - \alpha \right] &\text{if } p < 1,\\[0.5cm]
1 - \alpha \, \log \left(1 + \dfrac{1}{\alpha}\right) &\text{if } p = 1.
\end{array}	\right. \]
\end{theorem}

Both Daley and Kendall~\cite[Section~7]{DK} and Daley and Gani~\cite[Section~5.2]{DG} present the deterministic analysis of the \apdk.
Lebensztayn et al.~\cite[Example~2.8]{LTRM} examine the stochastic version, assuming also that an ignorant individual has a certain probability~\(q\) of becoming a stifler immediately upon hearing the rumor.
All these references deal with the asymptotic behavior of the final proportion of ignorants.
Recalling that the \dk is obtained by choosing $\al = p = 1$, we get the following result.\\

\begin{corollary}
\label{C: dk}
For the classical \dk,
\[ \lim_{N \to \infty} \pmax = 1 - \log 2 \approx 0.3069 \quad \alsu \]
\end{corollary}

\subsection{The \((\texorpdfstring{\alpha}{alpha}, p)\) Maki--Thompson model}
\label{SS: apmt}

The definition of the \emph{\vapmt}, or \emph{\apmt} for short, runs along similar lines, but with the rumor being transmitted by directed contacts of spreaders, and, in rule (b), only the initiating spreader having probability~\(\al\) to become a stifler.
Hence, the Markov process \(\{(X(t), Y(t))\}_{t \geq 0}\) evolves according to the following transition rates:

\begin{equation}
\label{F: Rates apmt}
\begin{array}{ccc}
\text{interaction} \qquad &
\text{transition} \qquad &\text{rate}\\[0.2cm]
\text{spreader -- ignorant} \qquad &
(-1, 1) \qquad &p \, X Y,\\[0.2cm]
\text{spreader -- spreader/stifler} \qquad &(0, -1) \qquad &p \, \al \, Y (N - X).
\end{array}
\end{equation}

The first equation of \eqref{F: Rates apmt} represents the case in which a spreader contacts an ignorant and chooses to transmit the rumor, which occurs with probability $p$; and as a result, the ignorant becomes a spreader. The second equation of \eqref{F: Rates apmt} corresponds to the case in which a spreader contacts a non-ignorant and chooses to transmit the rumor---an event that occurs with probability $p$, and after the transmission, the contacting spreader becomes a stifler, which happens with probability $\alpha$.

\begin{theorem}
\label{T: apmt}
Consider the \apmt with any \(\alpha, p \in (0, 1]\), and let \(\pmax\) be as in \eqref{F: PMax}. Then,
\[ \lim_{N \to \infty} \pmax = \lmax{MT}(\alpha) \quad \alsu, \]
where
\[ \lmax{MT}(\alpha) = 1 - \alpha \, \log \left(1 + \frac{1}{\alpha}\right). \]
\end{theorem}

By making $\al = p = 1$, we recover the classical model.

\begin{corollary}
\label{C: mt}
For the classical \mt,
\[ \lim_{N \to \infty} \pmax = 1 - \log 2 \approx 0.3069 \quad \alsu \]
\end{corollary}

\subsection{Discussion}
\label{SS: Discussion}

\begin{figure}[ht]
\centering
\includegraphics[scale=0.9]{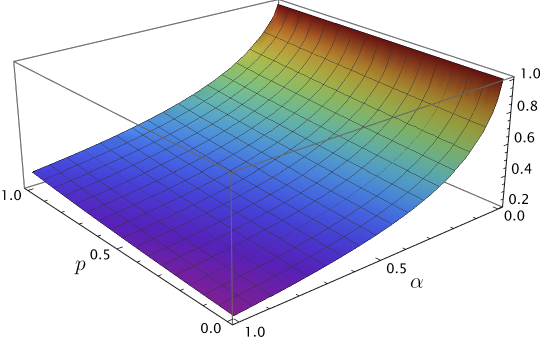}
\caption{Limiting proportion of the rumor peak for the \apdk.}
\label{Fig: apdk}
\end{figure}

\begin{figure}[ht]
\centering
\includegraphics[scale=0.9]{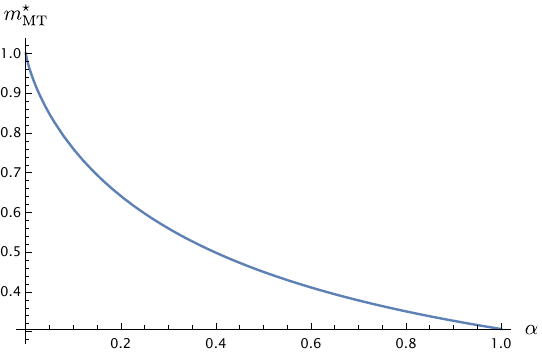}
\caption{Limiting proportion of the rumor peak for the \apmt.}
\label{Fig: apmt}
\end{figure}

Figures \ref{Fig: apdk} and \ref{Fig: apmt} show the graphs of \(\lmax{DK}\ap\) and \(\lmax{MT}(\alpha)\), respectively. Notice that, for the \apdk with any \(p\), the asymptotic rumor peak \(\lmax{DK}\) is a decreasing function of~\(\alpha\), which makes intuitive sense.
In addition,
\begin{equation}
\label{F: Lim Hayes}
\inf_{\ap \in (0, 1]^2} \lmax{DK}\ap = \lim_{p \to 0^{+}} \lmax{DK}(1, p) = \frac{2}{\sqrt{e}} - 1 \approx 0.2131,
\end{equation}
and for every \(p \in (0, 1]\),
\[ \sup_{\ap \in (0, 1]^2} \lmax{DK}\ap = \lim_{\alpha \to 0^{+}} \lmax{DK}\ap = 1. \]
We observe that the value in~\eqref{F: Lim Hayes} obtained for the \apdk with \(\alpha = 1\) and \(p \to 0^{+}\) corresponds to the asymptotic peak for a version of the \mt in which whenever a spreader interacts with another spreader, both become stiflers.
This model was formulated in Hayes~\cite{Hayes} by mistake when trying to simulate the basic \dk.

Recall that in \mt the contacts between individuals are directed, that is, the initiator and the recipient of the communication are distinguished.
On the other hand, in \dk the interactions are pairwise.
Thus, considering the \(\ap\) versions of these models, we see that in~\eqref{F: Rates apmt} the probability~\(p\) plays the role of a time constant, while in~\eqref{F: Rates apdk} it does not.
Consequently, in the case of the \apmt, the limiting maximum proportion \(\lmax{MT}\) does not depend on~\(p\).
Of course, \(\lmax{MT}\) is a decreasing function of~\(\alpha\). For the classical models, obtained with the choice \(\al = p = 1\), the asymptotic rumor peak values match.
This property is reminiscent of what occurs with the final proportions of ignorants remaining in the population (for both classic models, they are equal to~\(0.2032\)).

\subsection{A general rumor model}
\label{SS: GRM}

To prove the previous results, we establish a limit theorem for the random variable \(\pmax\) in a stochastic rumor model that includes as particular cases the \apdk and the \apmt.
This \emph{general rumor model} is described in terms of parameters so that a meeting between a spreader and an ignorant does not always produce a new spreader, and an encounter between two spreaders yields a random number of stiflers.
With the notation introduced above, i.e. $X(t)$ and $Y(t)$ denoting the number of ignorants and spreaders at time $t$, respectively, we consider the general rumor model defined by the continuous-time Markov chain \(\{(X(t), Y(t))\}_{t \geq 0}\)  with initial state \((N, 1)\) and
\begin{equation*}
{\allowdisplaybreaks
\begin{array}{ccc}
\text{interactions} \quad &\text{transition} \quad &\text{rate}\\[0.2cm]
\text{spreader -- ignorant}\quad &(-1, 1) \quad &\la \, X Y,\\[0.2cm]
\text{spreader -- spreader}\quad&(0, -2) \quad &\la \, \tu \displaystyle\binom{Y}{2},\\[0.4cm]
\text{spreader -- spreader/stifler}\quad&(0, -1) \quad &\la \, \td \, Y (Y - 1) + \la \, \ga \, Y (N + 1 - X - Y).
\end{array}}%
\end{equation*}
The parameter \(\la\) represents a meeting rate, and the first transition corresponds to the transformation of an ignorant individual into a spreader after hearing the rumor.
The second case indicates the transition in which two spreaders meet and both become stiflers.
The third transition occurs in an interaction spreader--spreader or spreader--stifler when only one spreader turns into a stifler.
The general rumor model was proposed by Lebensztayn et al.~\cite{LTRM}, who proved a Law of Large Numbers and a Central Limit Theorem for the fractions of individuals in the different classes of the population at the end of the process.

Let $\te = \tu + \td - \ga$.
We assume henceforth that
\(\la, \ga \in (0, \infty)\), \(\tu, \td \in [0, \infty)\), and \(\te \in [0, 1]\).
Notice that the \apdk is obtained by making
\begin{equation}
\label{F: Values apdk}
\la = p, \, \tu = \al^2 (2 - p), \, \td = \al (1 - \al) (2 - p), \, \ga = \al \, \text{ and } \, \te = \al (1 - p),
\end{equation}
whereas for the \apmt,
\begin{equation}
\label{F: Values apmt}
\la = p, \, \tu = 0, \, \td = \ga = \al \, \text{ and } \, \te = 0.
\end{equation}


\begin{theorem}
\label{T: gm}
Consider the general rumor model with parameters $\la, \ga \in (0, \infty)$, $\tu, \td \in [0, \infty)$, $\te \in [0, 1]$, and let \(\pmax\) be defined as in \eqref{F: PMax}. Then,
\[ \lim_{N \to \infty} \pmax = \lmax{GM}\gt \quad \alsu, \]
where
\[ \lmax{GM}\gt =
\left\{
\begin{array}{cl}
1 - \gamma \, \log \left(1 + \dfrac{1}{\gamma}\right) &\text{if } \te = 0,\\[0.5cm]
\dfrac{1}{\theta \, (1 - \theta)} \left[ (\theta + \gamma) \left(\dfrac{\theta + \gamma}{1 + \gamma}\right)^{\frac{\theta}{1 - \theta}} - \theta \, (1 + \gamma) \left(\dfrac{\theta + \gamma}{1 + \gamma}\right)^{\frac{1}{1 - \theta}} \right] - \dfrac{\gamma}{\theta} &\text{if } 0 < \te < 1,\\[0.6cm]
(1 + \gamma) \, \exp\left(-\dfrac{1}{1 + \gamma}\right) - \gamma &\text{if } \te = 1.
\end{array}	\right. \]
\end{theorem}

We remark that \(\lmax{GM}\gt\) is a continuous function for \(\gt \in (0, \infty) \times [0, 1]\).

\section{Proofs}
\label{S: Proofs}

First, we observe that Theorems \ref{T: apdk} and \ref{T: apmt} are direct consequences of Theorem~\ref{T: gm}, by replacing in \(\lmax{GM}\) the appropriate values of \(\gamma\) and \(\theta\) that are given in equations~\eqref{F: Values apdk} and \eqref{F: Values apmt}, respectively.

To prove Theorem~\ref{T: gm}, we consider the general rumor model and define
\[ V^{(N)}(t) = (X^{(N)}(t), Y^{(N)}(t)), t \geq 0.\]
Here are the main ideas of the proof.
Through a random time change, we define a new process $\{\tilde{V}^{(N)}(t)\}_{t \geq 0}$ that has the same transitions as $\{V^{(N)}(t)\}_{t \geq 0}$, whence their corresponding maximum values of spreaders until absorption coincide.
The coupling between $\{V^{(N)}(t)\}_{t \geq 0}$ and $\{\tilde{V}^{(N)}(t)\}_{t \geq 0}$ is done in a way that the latter process is a density-dependent Markov chain, for which the theory presented in Ethier and Kurtz~\cite[Chapter~11]{MPCC} can be used.
Thus, $\{\tilde{V}^{(N)}(t)\}_{t \geq 0}$ suitably rescaled converges almost surely to the solution of a system of differential equations, uniformly on bounded time intervals.
By studying this solution, we show that $\lmax{GM}$ is the maximum proportion of spreaders of the limiting deterministic trajectory.
Finally, by exploring the convergence of the coupled model to the solution of the dynamical system, we prove that the maximum proportion of spreaders in $\{\tilde{V}^{(N)}(t)\}_{t \geq 0}$ converges almost surely to $\lmax{GM}$.

\subsection{Coupled density-dependent stochastic model}
\label{SS: Coupled SM}

For a realization of the general rumor model $\{V^{(N)}(t)\}_{t \geq 0}$, we define
\begin{align*}
\Theta^{(N)}(t) &= \int_0^t Y^{(N)}(s) \, ds \text{ for } 0 \leq t \leq \tabs,\\
\Upsilon^{(N)}(s) &= \inf \{t: \Theta^{(N)}(t) > s\} \text{ for } 0 \leq s \leq \int_0^{\infty} Y^{(N)}(u) \, du,
\end{align*}
and let $ \tilde{V}^{(N)}(t) = V^{(N)}(\Upsilon^{(N)}(t)) $.
Then the time-changed process $ \{\tilde{V}^{(N)}(t)\}_{t \geq 0} $ has the same transitions as $ \{V^{(N)}(t)\}_{t \geq 0} $.

We consider
\begin{equation*}
\ttabs = \inf \{t: \tilde{Y}^{(N)}(t) = 0\}
\end{equation*}
the corresponding absorption time, and let
\[ \tilde{v}^{(N)}(t) = \dfrac{\tilde{V}^{(N)}(t)}{N} = (\tilde x^{(N)}(t), \tilde y^{(N)}(t)), t \geq 0, \]
be the rescaled time-changed process.
Notice that
\[ \tpmax = \max_{0 \leq t \leq \ttabs} {\tilde{y}^{(N)}(t)}  = \max_{0 \leq t \leq \ttabs} \frac{\tilde{Y}^{(N)}(t)}{N} \]
satisfies
\begin{equation}
\label{F: Max}
\tpmax = \pmax.
\end{equation}

We underline that $ \{\tilde{v}^{(N)}(t)\}_{t \geq 0} $ is a \emph{density-dependent Markov family}.
Indeed, note that $ \{\tilde{V}^{(N)}(t)\}_{t \geq 0} $ is a continuous-time Markov chain with state space $\bbZ^2$, and transition rates given by
{\allowdisplaybreaks
\begin{equation}
\label{F: Rates TCP}
\begin{array}{cc}
\text{transition} \quad &\text{rate}\\[0.1cm]
\ell_0 = (-1, 1) \quad &\la \, \tilde X,\\[0.2cm]
\ell_1 = (0, -2) \quad &\la \, \tu \displaystyle\frac{\tilde Y - 1}{2},\\[0.28cm]
\ell_2 = (0, -1) \quad &\la \, \td \, (\tilde Y - 1) + \la \, \ga \, (N + 1 - \tilde X - \tilde Y).
\end{array}
\end{equation}}%
Now defining
\begin{align*}
\beta_{\ell_0} (x, y) &= \la \, x,\\[0.1cm]
\beta_{\ell_1} (x, y) &= \la \, \tu \, \frac{y}{2},\\[0.1cm]
\beta_{\ell_2} (x, y) &= \la \, \td \, y + \la \, \ga \, (1 - x - y),
\end{align*}
we see that the transition intensities in~\eqref{F: Rates TCP} can be expressed as
\[ N \left[ \beta_{\ell_i} \left( \dfrac{\tilde X}{N}, \dfrac{\tilde Y}{N} \right) + O \left( \dfrac{1}{N} \right) \right]. \]
Therefore, $ \{\tilde{v}^{(N)}(t)\}_{t \geq 0} $ is a density-dependent population process with possible transitions in the set $\{\ell_0, \ell_1, \ell_2\}$.

\subsection{Limiting dynamical system}
\label{SS: Dynamical system}

We now apply Theorem~11.2.1 from Ethier and Kurtz~\cite{MPCC} to show that the time-changed model converges almost surely as $N \to \infty$, uniformly on compact time intervals.
Here, the \emph{drift function} is given by
\[ F(x, y) = \sum_{i=0}^{2} \ell_i \, \beta_{\ell_i} (x, y) 
= (- \la x, \la (1 + \ga) x - \la \te y - \la \ga). \]
Consequently, the corresponding deterministic system is defined by the following set of differential equations:
\begin{equation}
\label{F: SDE}
\begin{cases}
x^{\prime}(t) = - \la \, x(t),\\[0.1cm]
y^{\prime}(t) = \la (1 + \ga) \, x(t) - \la \te \, y(t) - \la \ga,
\end{cases}
\end{equation}
with initial state $(x(0), y(0)) = (1, 0)$.

To provide an explicit formula for the solution, we define the function $f: (0, 1] \rightarrow \bbR$ given by
\begin{equation}
\label{F: Function f}
f(x) = 
\left\{
\begin{array}{cl}
(1 + \ga) (1 - x) + \ga \log x &\text{ if } \, \te = 0,\\[0.2cm]
\dfrac{(\ga + \te) x^{\te} - (1 + \ga) \te x - \ga (1 - \te)}{\te (1 - \te)} &\text{ if } \, 0 < \te < 1,\\[0.3cm]
- \ga (1 - x) - (1 + \ga) \, x \log x &\text{ if } \, \te = 1.
\end{array}	\right.
\end{equation}
Then, the solution of system~\eqref{F: SDE} is $v(t) = (x(t), y(t))$, where
\begin{equation}
\label{F: Sol}
x(t) = e^{-\la t}
\quad \text{and} \quad
y(t) = f(x(t)).
\end{equation}
As stated in Theorem~11.2.1 of Ethier and Kurtz~\cite{MPCC},

\begin{lemma}
\label{L: Conv v}
For every $T > 0$, we have that
\[ \lim_{N \to \infty} \sup_{t \leq T} \, \bigl|\tilde{v}^{(N)}(t) - v(t)\bigr| = 0 \quad \alsu \]
\end{lemma}

\subsection{Analysis of the deterministic system and convergence of the absorption times}
\label{SS: Analysis}

We begin this section with an analysis of the trajectory of the dynamical system solution.

\begin{lemma}
\label{L: Max}
For every $\ga > 0$ and $0 \leq \te \leq 1$, the function $f$ has a unique zero $\xf$ in the interval $(0, 1)$.
Now define
\begin{equation}
\label{F: tf}
\tf = -\frac{1}{\la} \, \log \xf.
\end{equation}
Then,
\begin{equation}
\label{F: Max y}
\max_{0 \leq t \leq \tf} y(t) = \lmax{GM}\gt.
\end{equation}
\end{lemma}

\begin{proof}
For \(0 < \te < 1\), we have that $\lim_{x \to 0^{+}} f(x) < 0$, $f(1) = 0$, and the derivative of $f$ with respect to $x$ is given by
\[ f^{\prime}(x) = \dfrac{(\gamma + \theta) \, x^{\theta - 1} - (1 + \gamma)}{1 - \theta}. \]
Hence, $f$ is a unimodal function, with the global maximum at the point
\[ x_{\max}\gt = \left(\frac{\ga + \te}{1 + \ga}\right)^{\frac{1}{1 - \te}} \in (0, 1). \]
Therefore, the function $f$ has a unique zero $\xf$ in the interval $(0, 1)$.
With~\eqref{F: Sol} in mind, we conclude that $t \mapsto x(t)$ establishes a one-to-one correspondence between the intervals $[0, \infty)$ and $(0, 1]$, and that $\tf$ defined in \eqref{F: tf} satisfies
\[ \tf = \inf \{t: y(t) \leq 0\}. \]
Furthermore,
\[ \max_{0 \leq t \leq \tf} y(t) = f(x_{\max}\gt) = \lmax{GM}\gt. \]

The same line of reasoning applies in the cases $\te = 0$ and $\te = 1$, with the function $f$ given in~\eqref{F: Function f},
\[ x_{\max}(\gamma, 0) = \left(1 + \dfrac{1}{\gamma }\right)^{-1},
\quad \text{and} \quad
x_{\max}(\gamma, 1) = \exp\left(-\dfrac{1}{1 + \gamma}\right). \qedhere \]
\end{proof}

As we show in Section~\ref{SS: Proof}, Lemma~\ref{L: Max} plays an essential role in the proof of Theorem~\ref{T: gm}, since formula~\eqref{F: Max y} is a key ingredient for the argument.  The next lemma provides another significant result, establishing that the sequence of the absorption times for the coupled process converges almost surely to $\tf$.

\begin{lemma}
\label{L: Conv tf}
We have that
\[ \lim_{N \to \infty} \ttabs = \tf \quad \alsu \]
\end{lemma}

\begin{proof}
Notice that $y^{\prime}(0) > 0$ and $y^{\prime} (\tf) = \la (1 + \ga) \, \xf - \la \ga < 0$, whence $y(\tf - \eps) > 0$ and $y(\tf + \eps) < 0$ for every $0 < \eps < \tf$.
The result follows from Lemma~\ref{L: Conv v}.
\end{proof}

\begin{remark}
As proved in Lebensztayn et al.~\cite{LTRM},
\[ \lim_{N \to \infty} \dfrac{X^{(N)}(\tau^{(N)})}{N} = x(\tf) = \xf \quad \alsu \]
In other words, \(\xf\) is the limiting proportion of ignorant individuals remaining in the population.
For both Daley--Kendall and Maki--Thompson models, we have that
\[ \xf = - \dfrac{\LW{- 2 \, e^{-2}}}{2} \approx 0.2032, \]
where $W_0$ is the principal branch of the Lambert $W$ function.
More details on this function can be found in Corless et al.~\cite{LW}.
\end{remark}

\subsection{Proof of Theorem \ref{T: gm}}
\label{SS: Proof}

In view of formula~\eqref{F: Max}, it suffices to show that
\begin{equation}
\label{F: Conv Max}
\lim_{N \to \infty} \tpmax = \lmax{GM}\gt \quad \alsu
\end{equation}
Let \((\Omega, \mathcal{F}, P)\) denote the probability space where both the processes \(\{V^{(N)}(t)\}_{t \geq 0}\) and \(\{\tilde{V}^{(N)}(t)\}_{t \geq 0}\) are defined.
When necessary, we will make explicit the dependence on \(\omega\) in the notation of random variables defined on \(\Omega\); for instance, we write \(\ttabs(\omega)\) instead of \(\ttabs\).
Also, let $ \mathbb{Q}_{+}$ be the set of positive rational numbers.
To prove~\eqref{F: Conv Max}, we first note that for any \(\omega \in \Omega\) and \(N \geq 3\), the function \(u \mapsto \max_{0 \leq t \leq u} \tilde{y}^{(N)}(t, \omega)\) is nondecreasing. 
Moreover,
\begin{equation}
\label{F: yabs}
\tilde{y}^{(N)}(t) = \tilde{y}^{(N)}(\ttabs) \text{ for all } t \geq \ttabs.
\end{equation}

Now we define the events
\begin{align*}
E_k &= \left\{\omega \in \Omega: \lim_{N \to \infty} \max_{0 \leq t \leq k} \tilde{y}^{(N)}(t, \omega) = \max_{0 \leq t \leq k} y(t)\right\}, k \in \bbQP,\\[0.1cm]
E &= \left\{\omega \in \Omega: \lim_{N \to \infty} \ttabs(\omega) = \tf\right\},\\[0.1cm]
E^{\ast} &= E \cap \bigl(\bigcap_{k \in \bbQP} E_k\bigr).
\end{align*}
From Lemmas~\ref{L: Conv v} and \ref{L: Conv tf}, we have that \(P(E) = 1\) and \(P(E_k) = 1\) for every \(k \in \bbQP\).
Hence \(P(E^{\ast}) = 1\), and it remains to establish that
\begin{equation}
\label{F: Imp}
\omega \in E^{\ast} \Rightarrow \lim_{N \to \infty} \max_{0 \leq t \leq \ttabs(\omega)} \tilde{y}^{(N)}(t, \omega) = \lmax{GM}\gt.
\end{equation}

Let \(\omega \in E^{\ast}\), and choose \(k \in \bbQP\) such that \(k > \tf\). 
Since \(\omega \in E^{\ast} \subset E\), it follows that there exists \(N_0\) such that \(k > \ttabs(\omega)\) whenever \(N \geq N_0\).
Consequently, from~\eqref{F: yabs} and the fact that \(\omega \in E_k\), we conclude that
\[ \lim_{N \to \infty} \max_{0 \leq t \leq \ttabs(\omega)} \tilde{y}^{(N)}(t, \omega) = 
\lim_{N \to \infty} \max_{0 \leq t \leq k} \tilde{y}^{(N)}(t, \omega) = 
\max_{0 \leq t \leq k} y(t). \]
Then we take a sequence \(\{k_i\}_{i \geq 1}\) such that \(k_i \in \bbQP\) and \(k_i \searrow \tf\) as \(i \to \infty\). 
Thus, for every \(i \geq 1\),
\[ \lim_{N \to \infty} \max_{0 \leq t \leq \ttabs(\omega)} \tilde{y}^{(N)}(t, \omega) = \max_{0 \leq t \leq k_i} y(t) . \]
So, letting \(i \to \infty\) and using~\eqref{F: Max y} from Lemma~\ref{L: Max}, we obtain
\[ \lim_{N \to \infty} \max_{0 \leq t \leq \ttabs(\omega)} \tilde{y}^{(N)}(t, \omega) = \max_{0 \leq t \leq \tf} y(t) = \lmax{GM}\gt. \]
This finishes the verification of~\eqref{F: Imp}, which completes the proof of Theorem~\ref{T: gm}.

\section{Final Remarks}
\label{S: Final Remarks}

Since they were formulated, the Daley--Kendall and Maki--Thompson models have been the focus of much research into how a piece of information spreads, from a theoretical point of view. For homogeneously mixed populations, Lebensztayn et al.~\cite{LTRM,RPRS} showed that the theory of density-dependent Markov chains can be used to get strong results about how these random models behave as the population size gets bigger. 
Similar topics were examined in Grejo and Rodriguez~\cite{GR}, Oliveira and Rodriguez~\cite{OR}, and Rada et al.~\cite{MRR}, for rumor and related models. In this work, we have taken a step forward in understanding the basic models by proving a limit theorem for the maximum proportion of spreaders, i.e., the peak of the propagation. This has been achieved through a suitable application of the theory of convergence of density-dependent Markov chains. The advantage of using this theory is that we can obtain results that not only identify some asymptotic values of interest, like the maximum proportion studied here, but also, with a bit more work, can describe the random fluctuations between these values and those from the original stochastic process.

It is worth pointing out that classical rumor processes assume that every spreader--ignorant contact has perfect transmission, and every frustrated communication interaction certainly results in stifling. These models are significantly embellished by the $\ap$-probability variants, which provide a more accurate representation of human behavior in rumor spreading. Our formulation of a general stochastic rumor model offers a unified mathematical framework that includes the standard models and their $\ap$-probability counterparts. This framework also allows for flexibility, enabling the parameters to be adjusted according to different social contexts and rumor types. From a statistical perspective, developing general rumor systems is advantageous, as traditional models often do not reflect real-world data. Nevertheless, introducing new models presents challenges, such as dealing with high-dimensional parameter spaces, managing constraints, and requiring time-series data across multiple realizations. To estimate the involved parameters, classical statistical techniques can be employed, such as methods based on martingale arguments and maximum likelihood procedures like the ones proposed by \cite{fierro} for the classical SIR epidemic model.


\section*{Acknowledgments}
The authors gratefully acknowledge the thoughtful comments and suggestions of two anonymous referees, which helped us to improve the paper.
This study was financed, in part, by the Conselho Nacional de Desenvolvimento Científico e Tecnológico (CNPq), Grant 316121/2023-1, the Funda\c{c}\~ao de Amparo \`a Ci\^encia e Tecnologia do Estado de Pernambuco (FACEPE), Grant APQ-1341-1.02/22, and the Funda\c{c}\~ao de Amparo \`a Pesquisa do Estado de S\~ao Paulo (FAPESP), Brazil, Process Number \#2023/13453-5.

\section*{Declarations}

\noindent
{\bf Conflict of interest.} The authors have no financial or proprietary interests in any material discussed in this article.

\end{document}